\documentclass[11pt]{siamart250211}

\usepackage{setspace}  

\usepackage{amsmath}
\usepackage{amssymb}
\usepackage{adjustbox}
\usepackage{xspace}

\usepackage{xfrac} 
\usepackage{xcolor}

\usepackage{ifthen}

\newcommand\myrepeat[2]{%
    \begingroup
    \lccode`m=`#2\relax
    \lowercase\expandafter{\romannumeral\number\number#1 000}%
    \endgroup
}

\newcounter{stepsindent}
\newcounter{stepscounter}

\newcommand{\stepcomment}[1]{\hfill\textit{\footnotesize #1}}

\newenvironment{steps}{%
    \newcommand*{\step}[1][EMPTY]
    {\ifthenelse{\equal{##1}{EMPTY}} {\item} {\item[##1]}
    \phantom{\myrepeat{\value{stepsindent}}{.}}
    }
    \newenvironment{indented}{\addtocounter{stepsindent}{3}}{\addtocounter{stepsindent}{-3}}
    \begin{list}
    {
            {\small \arabic{stepscounter}}.
    }
    {
        \usecounter{stepscounter}
        \setlength{\labelwidth}{2em}
        \setlength{\labelsep}{-0.2em}
        \setlength{\itemsep}{0pt}
        \setlength{\parsep}{1pt}
        \setlength{\leftmargin}{0em}
        \setlength{\rightmargin}{0cm}
        \setlength{\itemindent}{1.5em}
    }
}{\end{list}}

\renewtheorem{theorem}{Theorem}[section]
\renewtheorem{corollary}{Corollary}[section]
\renewtheorem{lemma}{Lemma}[section]
\renewtheorem{definition}{Definition}[section]
\newtheorem{observation}{Observation}[section]

\newcommand{\E}{\mathbb E}

\newcommand{\R}{\mathbb R}

\newcommand{\Rp}{\R_{\scriptscriptstyle 0}}
\newcommand{\Rpinf}{\overline{\R\!}\,_{\scriptscriptstyle 0}}

\newcommand{\eps}{\epsilon}

\newcommand{\giv}[1][]{\,{#1 |}\,}

\DeclareMathOperator{\minimize}{minimize}

\DeclareMathOperator{\maximize}{maximize}

\DeclareMathOperator{\NONE}{\mathsf{none}}
\newcommand{\none}{{\ensuremath{\NONE}}\xspace}

\DeclareMathOperator{\SAMPLE}{\mathsf{sample}}
\newcommand{\sample}{{\ensuremath{\SAMPLE}}\xspace}

\DeclareMathOperator{\GREEDY}{\mathsf{greedy}}
\newcommand{\greedy}{{\ensuremath{\GREEDY}}\xspace}

\newcommand{\capped}{\tilde c}

\makeatletter
\newcommand*{\tran}{\bgroup\@ifstar{\mathpalette\@tran{\mkern-3.5mu}\egroup}{\mathpalette\@tran\relax\egroup}}
\newcommand*{\@tran}[2]{\setbox0=\hbox{\m@th$#1#2\intercal$}\raise\dp0\box0}
\makeatother

\newenvironment{LP}[1]{
    \begin{array}[t]{@{}l@{}}
    \displaystyle #1 \\
    \begin{aligned}
    \\[-10pt]}{
    \end{aligned}
    \end{array}
}

\newcommand{\qedhere}{}

\renewcommand{\theenumi}{\roman{enumi}}%

\title{An improved approximation algorithm for $k$-Median}
\author{Neal E. Young\thanks{
    University of California, Riverside
}}
\date{\today}

\begin{document}



\maketitle

\begin{abstract}
    We give a polynomial-time approximation algorithm for the (not necessarily metric) $k$-Median problem.
    The algorithm is an \emph{$\alpha$-size-approximation} algorithm
    for $\alpha < 1+2\ln(n/k)$.
    That is, it guarantees a solution having size at most $\alpha\times k$,
    and cost at most the cost of any size-$k$ solution.
    This is the first polynomial-time approximation algorithm to match
    the well-known bounds of $H_{\Delta}$ and $1+\ln(n/\textsf{opt})$
    for unweighted Set Cover (a special case) within a constant factor.
    It matches these bounds within a factor of 2.
    The algorithm runs in time $O(k\,m\ln(n/k)\log m)$,
    where $n$ is the number of customers and $m$
    is the instance size.
\end{abstract}

\section{Introduction}\label{sec: intro}
An instance of the \emph{$k$-Median} problem
is given by an edge-weighted bipartite graph $G=(U, W, E)$,
where $U$ is the set of \emph{centers},
$W$ is the set of \emph{customers},
and each center/customer pair $(i,j) \in E \subseteq U\times W$
has an associated \emph{cost} $c_{ij} \ge 0$,
which we interpret as the cost of assigning customer $i$ to center $j$.
The goal is to choose a set $C\subseteq U$ of $k$ centers
of minimum \emph{cost}, defined to be $c(C)=\sum_{j\in W} \min_{i\in C} c_{ij}$,
with the interpretation that each customer is assigned to its closest center in $C$.
Let $n=|W|$ and $m=|E|$.

An \emph{$\alpha$-size-approximate solution}
is a set $C\subseteq W$ of size at most $\alpha\, k$
and cost at most the minimum cost of any size-$k$ solution.
An algorithm that guarantees such a solution
is an \emph{$\alpha$-size-approximation algorithm}.

The restriction of $k$-Median to zero-cost instances is equivalent to
the well-studied Set Cover problem---the Set Cover instance admits a cover of size $k$
if and only if the corresponding $k$-Median instance has a size-$k$ solution of cost zero.
Assuming P$\ne$NP,
this implies that
no polynomial-time approximation algorithm
for $k$-Median guarantees solutions of size $(1-\eps)\ln n$  (for any constant $\eps>0$)
with cost approximating the optimum within any finite factor~\cite{feige_threshold_1998}.
This motivates the study of polynomial-time bicriteria-approximation algorithms.

The first such algorithm, by Lin and Vitter,
produces solutions of size at most $(1+ 1/\eps)(1 + \ln n)k$
and cost at most $1+\eps$ times the minimum cost of any size-$k$ solution~\cite{lin_e_approximations_1992}.
Here $\eps>0$ is an input parameter that controls the tradeoff between size and cost.
The second algorithm, by Young, improves this tradeoff
by producing solutions of size at most $(1+\ln(n + n/\eps))k$
and cost at most $1+\eps$ times the minimum cost of any size-$k$ solution~\cite{young_k_medians_2000}.
The third, by Chrobak et al., incurs no tradeoff:
it guarantees solutions of size $O(\log n)$
and cost at most the minimum cost of any size-$k$ solution~\cite{chrobak_incremental_2008}.
That is, it is an $O(\log n)$-size-approximation algorithm as defined above.

\newcommand{\opt}{\textsc{opt}}
For the special case of unweighted Set Cover, stronger bounds are known
in terms of $k$ and the maximum set size $\Delta$.
Johnson~\cite{johnson_approximation_1974} and Lovasz~\cite{lovasz_ratio_1975} show
that the greedy algorithm has approximation ratio at most
$H_\Delta = \sum_{h=1}^\Delta 1/h$,
where $\Delta\le n$ is the maximum set size.
(Chvatal~\cite{chvatal_greedy_1979} extends this to weighted Set Cover.)
A folklore result (see \autoref{sec: poly time alg})
is that the ratio is at most
$1+\ln (n/\opt)$,
where $\opt$ is the optimal cover size.
Note that $n/\opt \le \Delta$ and $\gamma + \ln \Delta \le H_\Delta$
where $\gamma=0.577..$ is Euler's constant,
so this bound can be smaller than $H_\Delta$ but never exceeds it by more than $1-\gamma$.
(Slav\'{\i}k shows a bound of $\ln (n / \ln n) + O(1)$,
which is asymptotically stronger when $\opt = o(\log n)$~\cite{slavik_tight_1997}.)

\smallskip

Our main result (\autoref{cor: faster}) is a polynomial-time $\alpha$-size-approximation algorithm for $k$-Medians,
where $\alpha \le 1 + 2\ln(n/k) \le 2H_\Delta$.
(For $k$-Median $\Delta = \max\big\{|\{j\in W : (i,j)\in E\}| : i\in U\big\}$
is the maximum number of customers that any center can serve.
Note $n/k \le \Delta$ for any feasible instance.)
This matches the Set Cover bounds $H_\Delta$ and $1+\ln(n/\opt)$ within a factor of two.
No previous result matched either bound within any constant factor.

The only previous polynomial-time size-approximation algorithm~\cite{chrobak_incremental_2008}
requires solving the standard linear-program (LP) relaxation for $k$-Median
(\autoref{fig: k median LP}).
Our algorithm avoids that.
It runs in time $O(k\, m \log(n/k)\log m)$.

\newcommand{\SML}[1]{\text{\footnotesize \(#1\)}}

\newcommand{\kmedianLP}{
        {
        \begin{LP}{\minimize~ c\cdot y = \textstyle \sum_{ij} c_{ij}\, y_{ij}}
        \SML{(\forall i \in U,\, j \in W )} && x_i  \ge y_{ij} & {} \ge 0
        \\
        \SML{(\forall j \in W )} && \textstyle \sum_{i} y_{ij} & {} = 1
        \\
        && \textstyle \sum_{i} x_i & {} = k
        \end{LP}
    }
}

\newcommand{\kmedianDual}{
        {
        \begin{LP}{\maximize~ -k\,\mu + \textstyle \sum_{j}\delta_j}
        \SML{(\forall i\in U,\, j\in W)} && \pi_{ij} & \ge 0
        \\
        \SML{(\forall i\in U,\, j\in W)} && \pi_{ij} & \ge \delta_j - c_{ij}
        \\
        \SML{(\forall i\in U)} &&  \mu & \ge \textstyle \sum_{j} \pi_{ij}
        \end{LP}
    }
}

\begin{figure}
    \centering
    \begin{tabular}{c|c}
        \kmedianLP &
        \kmedianDual
    \end{tabular}
    \caption{
        The standard $k$-Median linear-program relaxation and its dual.
    }
    \label{fig: k median LP}
\end{figure}

\section{Preliminaries}\label{sec: preliminaries}
Let $\Rp = \{x \in \R : x \ge 0\}$ and $\Rpinf = \Rp\cup\{\infty\}$.

Fix a $k$-Median instance $G=(U,W, E)$.
For $(i,j)\in E$, let $c_{ij}\in \Rp$
denote the cost of assigning customer $j$ to center $i$.
To ease notation, take $c_{ij} = \infty$ for $(i,j)\in (U\times W) \setminus E$.

Fix  $T=\lceil k\ln(n^2/(2k(2k+1))\rceil$ and $\alpha_{kn} = T/k+2$,
and let $\lambda^*\in\Rpinf$ denote the optimum cost of the standard LP relaxation (\autoref{fig: k median LP}).
The algorithms here find solutions of cost at most $\lambda^*$
and size at most $T+2k = \alpha_{kn}\, k$.
These are $\alpha_{kn}$-size-approximate solutions,
because any size-$k$ solution has cost at least $\lambda^*$.

Assume without loss of generality that $2 \le k \le n/3$.
(If $k=1$, the optimal size-1 solution $C=\{\arg\min_{i\in U} \sum_{j\in W} c_{ij}\}$
can be computed in linear time.
If $k > n/3$, the size-$n$ solution $C=\{\arg\min_{i\in U} c_{ij} : j\in W\}$
has minimum possible cost and size $n \le \alpha_{nk} \,k$.)

Assume without loss of generality that every customer
$j\in W$ has a center $i\in U$ such that $c_{ij} = 0$.
(Otherwise,
for each customer $j$,  subtract $\min_{i'\in U} c_{i'j}$ from each $c_{ij}$.
This reduces the cost of each solution by the same non-negative amount,
$\sum_{j\in W} \min_{i\in U} c_{ij}$,
so any given solution achieves optimal cost for the original instance
if and only if it does so for the modified instance.)

\begin{lemma}
    \label{lemma: alpha}
    $\alpha_{kn} < 2\ln(n/k) + 2- 2\ln 2 + 1/(2k) + 1/(4k^2) < 1 + 2\ln (n/k) < 2H_{\Delta}$
\end{lemma}

The proof of the lemma is in \autoref{sec: appendix}.

The capped cost, defined next, plays a central role in the algorithms.
\autoref{sec: intuition} gives some intuition for its definition.

\begin{definition}[capped cost]
    \label{def: capped cost}
    Given $\lambda\in\Rpinf$ and a set $C\subseteq W$ of centers,
    define the \emph{capped cost} of $C$ to be
    $\capped(\lambda, C) = \sum_{j\in W} \capped_j(\lambda, C)$,
    where \[\capped_j(\lambda, C)=\min\Big(1/(2k+1), (1-2k/n)\min_{i\in C} c_{ij}/\lambda\Big)\]
    is the \emph{capped cost of customer $j$} (with respect to $C$).
    In this context we interpret $0/0$ and $\infty/\infty$ as zero.
\end{definition}

We'll use the following utility lemma to work with capped costs.
It shows that the addition of a single center $i'$ to a given set $C$ of customers
can always decrease the capped cost by a certain amount.
The lemma implicitly gives a lower bound on $\lambda^*$.

\begin{lemma}
    \label{lemma: submodular}
    For any
    $\lambda\in\Rpinf$ and
    $C\subseteq W$,
    \[
        \min_{i \in U} \capped(\lambda, C\cup\{i'\})
        \le  (1-1/k)\capped(\lambda, C) + (1/k)(1-2k/n)\lambda^*/\lambda.
    \]
\end{lemma}
We prove the lemma in the appendix.
For intuition
consider the case that the LP admits an optimal integer solution,
that is, a set $C^*\subseteq W$ of $k$ centers with cost $c(C^*) = \lambda^*$.
Adding all $k$ centers from $C^*$ to $C$ would reduce the capped cost
to at most $\capped(\lambda, C^*)$.
So, observing that $\capped(\lambda, C)$ (for fixed $\lambda$) is a submodular function,
there exists a center $i \in U$ to add
that reduces the capped cost by at least $(1/k)(\capped(\lambda, C) - \capped(\lambda, C^*))$.
With $\capped(\lambda, C^*) \le (1-2k/n)c(C^*)/\lambda = (1-2k/n)\lambda^*/\lambda$
and some algebra, this implies the lemma.
The full proof considers the optimal LP solution $(x^*, y^*)$ instead of $C^*$.

\begin{table}
    \small
    \label{table: notation}
    \begin{tabular}{r@{ -- }lr@{ -- }l}
        $\Rp$ & the non-negative reals
        &
        $\Rpinf$ & $\Rp\cup\{\infty\}$
        \\
        $c$ & costs for $k$-Median instance
        & $i, j$ & center $i \in U$, customer $j \in W $
        \\
        $\lambda^*$ & optimum LP cost
        & $n, m$ & $n=|W|$, $m$ = $|E|$
        \\
        $\lambda$ & capped-cost parameter
        &
        $C$ & set of centers
        \\
        $\capped(\lambda, C)$
        & capped cost of $C$
        &
        $\capped_j(\lambda, C)$
        & capped cost of customer $j$ (w.r.t.~$C$)
        \\[5pt]
        $T$ &
        number of centers from first phase: \hspace*{-8em}
        & \multicolumn{1}{r@{${}={}$}}{$T$}
        & \multicolumn{1}{@{}l}{$\lceil k\ln(n^2/(2k(2k+1))\rceil$}
        \\[2pt]
        $\alpha_{kn}$ &
        size-approximation ratio:
        & \multicolumn{1}{r@{${}={}$}}{$\alpha_{kn}$}
        & \multicolumn{1}{@{}l}{$T/k+2$} 
    \end{tabular}
    \caption{Notation}
\end{table}

\section{Slow algorithm}\label{sec: poly time alg}
This section describes a slow algorithm.
\autoref{sec: faster alg} will build on it to prove the main result.
The approach is similar in spirit to the folklore result that,
for unweighted Set Cover,
the greedy algorithm gives $(1+\ln(n/k))$-approximation,
where $k=\opt$ is the minimum set-cover size.
This can be shown as follows: each set chosen by the greedy algorithm covers
at least a $1/k$ fraction of the remaining elements, so after any iteration $t$
at most $n(1-1/k)^t < n e^{-t/k}$ elements remain uncovered.
In particular, after $t=\lceil {k\ln(n/k)}\rceil$ iterations
less than $n e^{-t/k} = k$ elements remain.
Each subsequent iteration covers at least 1 element,
so there are at most $k-1$ additional iterations,
for a total of at most $(\ln(n/k) + 1) k$.

Similarly, the $k$-Median algorithms here have two phases.
The first phase chooses up to $T=\lceil k\ln(n^2/(2k(2k+1)))\rceil$ centers greedily,
minimizing the capped cost with each choice.
We show (using \autoref{lemma: submodular})
that this generates a set $C$ whose capped cost is strictly less than 1.
The second phase then ``polishes'' the partial solution,
adding up to $2k$ additional centers,
each chosen greedily to reduce the capped cost of a particular customer to zero.
This suffices to obtain a full solution (of size at most $T+2k$)
whose true cost is at most the optimum LP cost $\lambda^*$.

\label{sec: first phase}
The goal of the first phase is to compute a set $C$ of at most $T$ centers
with capped cost $\capped(\lambda^*, C)$ strictly less than 1:

\begin{theorem}[first phase]
    \label{thm: first phase}
    A set $C\subseteq W$ of size at most $T$
    with $\capped(\lambda^*, C) < 1$
    can be computed in polynomial time.
\end{theorem}
\begin{proof}
    The algorithm solves the LP to obtain $\lambda^*$,
    then chooses centers greedily,
    in each step decreasing the capped cost $\capped(\lambda^*, C)$
    as much as possible,
    just until the capped cost is less than 1:
    \begin{steps}
        \step solve the LP to obtain the optimum cost $\lambda^*$
        \step let $C\gets \emptyset$
        \step while $\capped(\lambda^*, C) \ge 1$:
        \begin{indented}
            \step let $i'=\arg\min_{i \in U} \capped(\lambda^*, C\cup\{i\})$
            \step let $C\gets C \cup\{i'\}$
        \end{indented}
        \step return $C$
    \end{steps}

    \smallskip

    Let $C_t$ denote the set $t$ at the end of each iteration $t$.
    Let $C_0=\emptyset$.
    Let $p = 1-1/k$.
    From \autoref{lemma: submodular} (with $\lambda=\lambda^*$), and using $\capped(\lambda,\emptyset)=n/(2k+1)$,
    it follows inductively that the algorithm maintains the invariant
    \begin{equation}
        \label{eq: invariant}
        \capped(\lambda^*, C_t)
        \le p^t \,n/(2k+1) + (1-p^t)(1-2k/n).
    \end{equation}

    Thus, the loop cannot iterate more than $T=\lceil k \ln(n^2/(2k(2k+1)))\rceil$ times,
    because if it reaches iteration $T$,
    then, after that after that iteration,
    using $p^T < \exp(-T/k) \le 2k(k+1)/n^2$,
    the invariant implies
    \[
        \capped(\lambda^*, C_T)
        < \frac{2k(2k+1)}{n^2}\times \frac{n}{2k+1} + 1 - \frac{2k}n = 1.
        \qedhere
    \]
\end{proof}

\label{sec: polish}
The first phase (by \autoref{thm: first phase})
computes a set $C\subseteq W$ of size at most $T$
such that the capped cost $\capped(\lambda^*, C)$ is strictly less than $1$.
Given this $C$,
the second phase
computes the desired $\alpha_{nk}$-size-approximate solution:

\begin{theorem}[second phase]
    \label{thm: second phase}
    Given any set $C\subseteq U$ of at most $T$ centers, and
    $\lambda\in\Rpinf$ such that $\capped(\lambda, C)<1$,
    a set $C'\subseteq U$ of size at most $T+2k$
    and cost at most $\lambda$
    can be computed in time $O(m + n\log n)$.
\end{theorem}

\begin{corollary}
    \label{cor: poly time}
    $K$-Median admits a polynomial-time $\alpha_{nk}$-size-approximation algorithm.
\end{corollary}

Before we prove the theorem,
note for intuition that $C$ can have at most $2k$ customers with $\capped_j(\lambda, C) = 1/(2k+1)$,
simply because each unassigned customer contributes $1/(2k+1)$ to $\capped(\lambda, C)$,
which is less than 1.
So by adding one center $i$ with $c_{ij}=0$ for each unassigned customer,
we could assure that all customers are assigned.
Similarly, by adding such centers for the $2k$ customers with maximum capped cost,
we could reduce the capped cost by a factor of $1-2k/n$, to less than $1-2k/n$,
which (by inspection of the capped cost) is just enough to ensure that the true cost is at most $\lambda$.
Naively, accomplishing both of above goals would take $4k$ additional centers.
The theorem shows that to accomplish both goals it suffices to add just $2k$ centers.

\begin{proof}[Proof of \autoref{thm: second phase}]
    Recall $k\le n/3$.
    Given $C$, compute $C'$ as follows:
    \begin{steps}
        \step while $c(C) > \lambda$:
        \begin{indented}
            \step let $j = \arg\max_{j\in W} \capped_j(\lambda, C)$
            \step let $C \gets C\cup \{i'\}$ where $i'=\arg\min_{i\in U} c_{ij}$
            \stepcomment{--- note: now $c_{i'j} = \capped_j(\lambda, C) = 0$}.
        \end{indented}
        \step return $C$
    \end{steps}

    This can be done in $O(m + n\log n)$ time
    by presorting the set $W$ of customers by decreasing $\capped(\lambda, C)$.
    To finish, we show that the loop iterates at most $2k$ times.
    Assume it iterates at least $2k$ times (otherwise we are done).

    \begin{observation}
        \label{obs: max ave}
        For any vector $b \in \Rp^n$ such that $\sum_j b_j < 1$,
        let $b'$ be $b$ with its $2k$ largest values replaced by zero.
        Then
        \[
            \max_{j \in W } b'_j < 1/(2k+1), \text{~~and~~}
            \sum_{j\in W} b'_j < 1 - 2k/n.
        \]
    \end{observation}
    (Indeed, the maximum value in $b'$ is the minimum of the $2k+1$ largest values in $b$,
    which is at most the average of those values,
    which is at most $\sum_{b=1}^n b / (2k+1) < 1/(2k+1)$.
    So $\max_j b'_j$ is bounded as claimed.
    Zeroing $2k$ random values in $b$ would decrease its sum by a factor of $1-2k/n$ in expectation.
    Zeroing the $2k$ largest values decreases its sum by at least that.
    So $\sum_j b'_j$ is bounded as claimed.)

    \begin{observation}
        Let $C_{2k}$ be $C$ after $2k$ iterations of the loop.
        Then
        \label{obs: polished capped cost}
        \[
            \max_{j \in W } \capped_j(\lambda, C_{2k}) < 1/(2k+1), \text{~and~~}
            \capped(\lambda, C_{2k}) < 1-2k/n.
        \]
    \end{observation}

    (To see this, let $C_0$ refer to $C$ as given, before the loop executes,
    and define $b$ in $\Rpinf^n$ by $b_j = \capped_j(\lambda, C_0)$.
    Let $b'$ be obtained from $b$ by replacing its $2k$ largest values by zero.
    By the definition of $C_{2k}$, we have $\capped_j(\lambda, C_{2k}) \le b_j'$.
    The observation follows follows from \autoref{obs: max ave},
    along with
    $\sum_j b_j = \capped(\lambda, C_0) < 1$,
    and $\max_{j \in W } \capped_j(\lambda, C_{2k}) \le \max_{j \in W } b'_j$,
    and $\capped(\lambda, C_{2k}) \le \sum_j b'_j$.)

    \smallskip

    Finally,
    by the definition of the capped cost,
    the bound on $\max_j \capped_j(\lambda, C_{2k})$ in \autoref{obs: polished capped cost}
    implies that $\capped(\lambda, C_{2k}) = (1-2k/n)c(C_{2k})$.
    This and the bound on $\capped(\lambda, C_{2k})$  in \autoref{obs: polished capped cost}
    imply $c(C_{2k}) \le \lambda$, so that the loop terminates after iteration $2k$.
    This proves \autoref{thm: second phase}.
\end{proof}

\section{Fast algorithm}\label{sec: faster alg}
The running time of the algorithm in \autoref{cor: poly time}
is dominated by the time to compute the optimal LP solution.
But the algorithm uses only a single parameter of that solution, namely its cost $\lambda^*$.
This section builds on that to give our main result---a faster algorithm.
It replaces the first phase by an algorithm
that, instead of solving the LP,
somehow computes a pair $(\lambda, C)$
such that $\lambda \le \lambda^*$
and $\capped(\lambda, C) < 1$:

\begin{theorem}
    \label{thm: faster}
    Given just the instance $c$,
    in
    $O(k\,m\log(n/k)\log m)$ time
    one can compute a $\lambda\le \lambda^*$
    and a set $C\subseteq U$ of size at most $T$
    such that $\capped(\lambda, C) < 1$.
\end{theorem}

With \autoref{thm: second phase}, this will give the following result:

\begin{corollary}
    \label{cor: faster}
    $K$-Median admits an $\alpha_{nk}$-size-approximation algorithm
    that runs in $O(k\,m\log(n/k)\log m)$ time.
\end{corollary}

\begin{figure}
    \begin{steps}
        \step[] {\hspace*{-1.8em}$\greedy'(c)$}
        \stepcomment{--- input: cost vector $c$}
        \step let $C = \emptyset$ and $\lambda_0 = 0$ and $t\gets 0$
        \step while $\capped(\lambda_{t}, C) \ge 1$:
        \begin{indented}
            \step let $t \gets t + 1$
            \step let $\tau \gets (1-1/k)\capped(\lambda_{t-1}, C) + (1/k)(1-2k/n)$
            \step let $\lambda_t = \min\big\{\lambda\ge \lambda_{t-1}
            : \min_{i \in U} \capped(\lambda, C\cup\{i\}) \le \tau\big\}$
            \step choose $i'$ such that $\capped(\lambda_t, C\cup\{i'\}) \le \tau$
            \step let $C \gets C \cup \{i'\}$
        \end{indented}
        \step return $(\lambda_t, C)$
    \end{steps}
    \caption{Faster replacement for first phase.}
    \label{fig: pure greedy}
\end{figure}

\begin{proof}[Proof of \autoref{thm: faster}]
    The algorithm to compute $(\lambda, C)$ is $\greedy'(c)$ in \autoref{fig: pure greedy}.

    \paragraph{Correctness.}
    Consider executing $\greedy'(c)$.
    Let $C_t$ denote $C$ at the end of iteration $t$ (and $C_0=\emptyset$).
    We'll show that the algorithm maintains the invariant
    \begin{equation}
        \label{eq: faster invariant}
        \lambda_t \le \lambda^* \text{ and }
        \capped(\lambda_t, C_t) \le p^t n/(2k+1) + (1-p^t)(1-2k/n)
    \end{equation}
    Note the similarity to \eqref{eq: invariant} in the proof of \autoref{thm: first phase}.

    Invariant~\eqref{eq: faster invariant} is true initially because $\lambda_0 = 0$
    and $\capped(\lambda_0, \emptyset) = n/(2k+1)$.

    Suppose the invariant holds just before a given iteration $t$.
    That is,
    \[\lambda_{t-1} \le \lambda^* \text{ and }
    \capped(\lambda_{t-1}, C_{t-1}) \le p^{t-1} n/(2k+1) + (1-p^{t-1})(1-2k/n).\]

    First we argue that $\lambda_t \le \lambda^*$.
    In the case that $\lambda_t = \lambda_{t-1}$, this follows from $\lambda_{t-1} \le \lambda^*$.
    Otherwise, consider any $\lambda$ in the half-open interval $[\lambda_{t-1}, \lambda_t)$.
    The algorithm's choice of $\lambda_t$
    (using that $\capped(\lambda, C)$ is non-increasing with $\lambda$)
    ensures
    \begin{align*}
        \min_{i \in U} \capped(\lambda, C_{t-1}\cup\{i\})
        & >
        (1-1/k)\capped(\lambda_{t-1}, C_{t-1}) + (1/k)(1-2k/n)
        \\
        & \ge (1-1/k)\capped(\lambda, C_{t-1}) + (1/k)(1-2k/n),
    \end{align*}
    which, with \autoref{lemma: submodular},
    implies that $\lambda^*/\lambda > 1$.
    So $\lambda<\lambda^*$ for all $\lambda$ in the non-empty interval $[\lambda_{t-1}, \lambda_t)$.
    The desired bound $\lambda^* \ge \lambda_t$ follows.

    To finish showing that~\eqref{eq: faster invariant} is maintained
    we bound $\capped(\lambda_t, C_t)$.
    Recall that $p=1-1/k$.
    The choices of $i'$ and $\tau$, and the invariant at time $t-1$, give
    \begin{align*}
        \capped(\lambda, C_t\cup\{i'\})
        & \le
        p\, \capped(\lambda_{t-1}, C_{t-1}) + (1-p)(1-2k/n)
        \\
        & \le p\big[p^{t-1}n/(2k+1) + (1-p^{t-1})(1-2k/n)\big]+ (1-p)(1-2k/n)
        \\
        & = p^t n/(2k+1) + (1-p^t)(1-2k/n).
    \end{align*}
    Hence invariant~\eqref{eq: faster invariant} is maintained.
    Now suppose for contradiction that there are more than $T$ iterations.
    After iteration $T$, the invariant implies that $\lambda_T \le \lambda^*$,
    and by the choice of $p=1-1/k<e^{-1/k}$ and $T\ge k\ln(n^2/(2k(2k+1))$ we have
    \[
        \capped(\lambda_T, C_T)
        \le e^{-T/k} \frac{n}{2k+1} + 1-\frac{2k}n
        < \frac{2k(2k+1)}{n^2}\times \frac{n}{2k+1} + 1-\frac{2k}n
        = 1,
    \]
    so the loop terminates after iteration $T$,
    contradicting that there are more than $T$ iterations.
    Hence the pair $(\lambda_t, C_t)$ returned by the algorithm
    is as claimed in \autoref{thm: faster}.

    \paragraph{Run time.}
    The loop makes at most $T=O(k\log(n/k))$ iterations.
    To show the claimed time bound,
    we show that each loop iteration can be implemented
    using $O(\log m)$ iterations of binary search,
    each of which takes $O(m)$ time.

    For each center/customer pair $(i,j)\in W\times U$ define
    $\beta_{ij} = (2k+1)(1-2k/n) c_{ij}$.
    Define $(\beta_1, \beta_2, \ldots, \beta_N)$
    to be the set $\{\beta_{ij} : i \in U, j \in W \} \cup\{\beta_0, \beta_{\infty}\}$
    of \emph{breakpoints}, in increasing order.
    Note that $\beta_1=0$, $\beta_N=\infty$, and $N\le m+2$.

    In a given loop iteration $t\le T$,
    compute $\lambda_t$ as follows.

    Given any $\lambda$,
    we can query the condition ``$\lambda_t \le \lambda$'' in $O(m)$ time,
    using that $\lambda_t \le \lambda$
    if and only if $\min_{i \in U} \capped(\lambda, C\cup\{i\}) \le \tau$.
    Using this, first check whether  $\lambda_t \le \lambda_{t-1}$.
    If it holds we have found $\lambda_t$ (as $\lambda_t = \lambda_{t-1}$),
    so assume $\lambda_t > \lambda_{t-1}$.
    Use $O(\log m)$ iterations of binary search
    (checking $\lambda_t \le \beta_\ell$ for some $\ell\in[N-1]$ in each of these iterations)
    to find $\ell\in[N-1]$ such that
    that $\beta_\ell < \lambda_t \le \beta_{\ell+1}$.

    The capped cost $\capped(\lambda, C\cup\{i\})$ implicitly defines a ``partial'' assignment
    that assigns each customer $j$ to its closest center $g_j = \arg\min_{g\in C\cup\{i\}} c_{gj}$ in $C\cup\{i\}$
    if $c_{g_j,j}(1-2k/n)/\lambda < 1/(2k+1)$, and otherwise leaves $j$ unassigned.
    The open interval $(\beta_\ell, \beta_{\ell+1})$ contains no breakpoints,
    so this assignment is the same for all $\lambda$ in this interval.
    Hence, letting $ u(i)$ be the number of customers not assigned by the assignment,
    and letting $c'(i)$ denote the total cost of just the assigned customers,
    for any $\lambda$ in $[\beta_\ell, \beta_{\ell+1}]$ the capped cost
    $\capped(\lambda, C\cup\{i\})$ equals $c'(i)(1-2k/n)/\lambda + u(i)/(2k+1)$.
    Hence,
    \begin{align*}
        \lambda_t
        &\textstyle = \min\big\{ \lambda \ge \lambda_{t-1} :
        \min_{i \in U} \capped(\lambda, C\cup\{i\}) \le \tau\big\}
        \\
        &\textstyle = \min\big\{ \lambda\ge \lambda_{t-1} :
        \min_{i \in U} c'(i)(1-2k/n)/\lambda + u(i)/(2k+1) \le \tau  \big\}
        \\
        &\textstyle = \min\big\{ \lambda\ge \lambda_{t-1}:
        \min_{i \in U} c'(i)(1-2k/n)/(\tau - u(i)/(2k+1)) \le \lambda \big\}
        \\
        &\textstyle = \max\big(\lambda_{t-1},\, \min_{i \in U} c'(i)(1-2k/n)/(\tau - u(i)/(2k+1))\big).
    \end{align*}

    After computing $u(i)$ and $c'(i)$ for all $i \in U$ in $O(m)$ time (total),
    the algorithm computes $\lambda_t$ (as the right-hand side above) in $O(m)$ time.
    (In the case $\min_{i \in U}  u(i)=n$,
    this gives $\lambda_t = \infty$.)
    Given $\lambda_t$, it then computes the center $i'$ to add in $O(m)$ time.
    This proves \autoref{thm: faster}.
\end{proof}

\section{The capped cost is a pessimistic estimator}\label{sec: intuition}
This section gives some intuition for the somewhat mysterious capped cost
$\capped(\lambda, C)$.
Briefly, it is a pessimistic estimator from the analysis of a natural random experiment.

\begin{definition}
    \label{def: partial assignment}
    A \emph{partial assignment} is a vector $a\in( U\cup\{\none\})^W$, with the interpretation
    that $a_j$ is the center assigned to customer $j\in W$, or $a_j=\none$ if $j$ has no assigned center.
    Let the \emph{cost} of $a$ be $c(a) = \sum_{j: a_j\ne\none} c_{a_j, j}$.
    Let $ u(a)=|\{j\in W : a_j = \none\}|$ denote the number of unassigned customers.
\end{definition}

Consider the partial rounding scheme in \autoref{fig: existence}.
It takes as input the instance $c$ and the optimal solution $(x^*, y^*)$ to the LP.
It rounds the fractional solution to a partial assignment
by sampling $T=\lceil k\ln (n^2/(2k(2k+1)))\rceil$ times
(with replacement) from the distribution $x^*/k$,
and, with each sampled center $i$, assigning (or reassigning) each customer $j \in W $
to $i$ with probability $y^*_j/x^*_i \le 1$.
In this way the rounding scheme maintains a partial assignment $a$.
Note that $a$ can assign a customer to a center that is not its closest open center.
It returns the partial assignment resulting from $T$ such iterations.

Let random variable $A$ be the partial assignment returned by $\sample(x^*, y^*)$.

\begin{figure}
    \begin{steps}
        \step[]{\hspace*{-1.5em}$\sample(c, x^*, y^*)$}
        \stepcomment{--- $(x^*, y^*)$ is optimal solution to the LP}
        \step assign $a_j\gets\none$ for $j\in W$
        \stepcomment{--- initialize partial assignment $a$}
        \step do the following steps $T$ times:
        \begin{indented}
            \step choose center $i \in U$ randomly from distribution $x^*/k$
            \step for each customer $j \in W $: with probability $y^*_{ij}/x^*_i$, reassign $a_j \gets i$
        \end{indented}
        \step return $a$
    \end{steps}
    \caption{Rounding to a partial assignment by sampling.}
    \label{fig: existence}
\end{figure}

\begin{lemma}
    \label{lemma: pre first loop}
    $\E[c(A)] \le \lambda^*$
    and
    $\E[ u(A)] < 2k(2k+1)/n.$
\end{lemma}

\begin{proof}[Proof sketch.]
    As shown in~\cite{young_k_medians_2000},
    for any center $i \in U$ and customer $j \in W $
    the probability that any single iteration assigns customer $j$ to some center is $1/k$.
    Also, given that $j$ is assigned to some center (in any iteration),
    the expected cost of the assignment is $\sum_i y^*_{ij} c_{ij}$,
    so
    \begin{align*}
        \E[c(A)]
        & \textstyle {} = \sum_{ij} \Pr[a_j \ne \none] \times \Pr[a_j = i \giv a_j \ne \none] \, c_{ij}
        \\
        & \textstyle {} \le \sum_{ij} \Pr[a_j = i \giv a_j \ne \none] \, c_{ij}
        \\
        & \textstyle {} = \sum_{ij} y^*_{ij} \,c_{ij}
        \textstyle {} \,=\, c\cdot y^* \,=\, \lambda^*.
    \end{align*}

    Meanwhile, the probability that a given customer $j \in W $ is never assigned is
    $(1-1/k)^T < \exp^{-T/k} \le 2k(2k+1)/n^2$.
    The desired bound $\E[ u(A)] < 2k(2k+1)/n$ follows by linearity of expectation.
\end{proof}

\begin{lemma}
    \label{lemma: existence precursor}
    With positive probability,
    $A$ has cost $c(A) \le \lambda^*/(1-2k/n)$
    and $ u(A) \le 2k$.
\end{lemma}
\begin{proof}[Proof sketch.]
    Assume $0 <\lambda^* < \infty$.
    (The other cases are easy to verify.)
    \begin{align}
        & \hspace*{-1em}\Pr\big[c(A) \ge \lambda^*/(1-2k/n) \text{ \,or\, }  u(A) \ge 2k+1\big]
        \notag\\[3pt]
        & {} \le\,
        \Pr[c(A) \ge \lambda^*/(1-2k/n)] ~+~  \Pr[ u(A) \ge 2k+1]
        && (\text{naive union bound})
        \notag\\
        & {} \le\,
        \E[c(A)] \frac{1-2k/n}{\lambda^*}
        ~+~
        \E[ u(A)] \frac 1 {2k+1}
        && (\text{Markov bound})
        \label{eq: tilde c bound}
        \\
        &{} <\, 1-2k/n ~+~ 2k/n \,=\, 1
        && (\text{\autoref{lemma: pre first loop}}).
        \notag\qedhere
    \end{align}
\end{proof}

The proof of the lemma bounds the probability of the two ``bad'' events
by the expectation of a pessimistic estimator,
then shows that the expectation of that pessimistic estimator is less than 1.
Instead of using \autoref{lemma: existence precursor} directly,
we work directly with the pessimistic estimator,
which is the capped cost $\capped(\lambda, A)$ for $\lambda=\lambda^*$.

In fact,
the greedy algorithm in the first phase of the slow algorithm
can be obtained by applying the method of conditional probabilities
to the random-sampling rounding scheme
(following the approach initiated in \cite{young_randomized_1995})
to find an outcome with $\capped(\lambda^*, A) < 1$.

\section{Conclusion}\label{sec: conclusion}
To conclude we remark on the LP dual solutions implicitly generated by the algorithm,
sketch how the algorithms compare to those in previous works,
and give some open problems.

\paragraph{Implicit LP dual solutions.}
As expected,
the algorithms presented here implicitly define solutions to the dual of the standard $k$-Median LP relaxation
(\autoref{fig: k median LP}).
The proven approximation ratios hold with respect to the dual solution cost.

Briefly, rewriting the bound in \autoref{lemma: submodular}, it is
\begin{align*}
    \lambda^*
    & \ge
    \frac{\lambda}{1-2k/n}
    \big(\capped(\lambda, C) -
    k \max_{i\in U} [\capped(\lambda, C) - \capped(\lambda, C\cup\{i\})]\big).
\end{align*}

This lower bound is equivalent to weak duality
for the feasible dual solution $(\delta, \pi,\mu)$ defined by
\begin{align*}
    \delta_j
    & \textstyle = \frac{\lambda}{1-2k/n} \capped_j(\lambda, C)
    = \min\big(\frac{\lambda}{(1-2k/n)(2k+1)}, \min_{i\in C} c_{ij}\big),
    \\ \pi_{\ij} & = \max(0, \delta_j - c_{ij}),
    \text{ and } \mu = \max_{i\in U} \textstyle\sum_j \pi_{ij}.
\end{align*}

Each iteration of the first phase,
the current pair $(\lambda, C)$ defines a feasible dual solution
whose cost is a lower bound on $\lambda^*$.
The cost of the final primal solution is at most
the maximum cost of any of these dual solutions,
which is in turn at most $\lambda^*$.
This can be shown by mechanically recasting the relevant invariants in terms of the dual solutions.

\paragraph{Comparison to previous works.}
The bicriteria-approximation algorithm for $k$-Median in~\cite[\S 6]{young_k_medians_2000}
takes as input an instance $c$, an $\eps>0$, and any upper bound $\lambda$ on the optimal fractional cost $\lambda^*$.
It returns a solution of cost at most $(1+\eps)\lambda$ of size at most $1 + k\ln(n + n/\eps)$.
(Note that the size is $\Omega(k \ln n)$ regardless of $\eps$.)
That algorithm is derived by derandomizing a natural random-sampling rounding scheme.

The first $O(\log n)$-size approximation algorithm, in~\cite[Theorem 5]{chrobak_incremental_2008},
requires as input the instance \emph{and} the optimal cost $\lambda^*$.
It calls the algorithm from~\cite{young_k_medians_2000} with $\eps=1/n$ and $\lambda=\lambda^*$
to obtain a set $C$ of centers, then returns $C\cup\{i\}$,
where $i$ is chosen to minimize the cost of $C\cup\{i\}$.
Assuming without loss of generality that every customer has a center with assignment cost zero,
this reduces the cost by at least a factor of $1-1/n$, reducing the cost below $\lambda^*$.

The first phase of our first algorithm can be obtained by derandomizing the rounding scheme
from~\cite{young_k_medians_2000},
but with respect to a different analysis,
so ours makes fewer iterations and minimizes a different function.
Then, instead of adding just one center to bring the cost down,
our polishing step (which is the key to obtaining the stronger approximation ratio) adds $2k$ centers.
Finally, all the previous size-approximation algorithms require solving the LP.
Our faster algorithm avoids this as described in \autoref{sec: faster alg}.

\paragraph{Open problems.}
Is there a polynomial-time algorithm with size-approximation ratio $\alpha \ln (n/k) + o(\log (n/k))$
for some constant $\alpha < 2$?

In the more general \emph{weighted} $k$-Median problem,
each center $i$ is given a weight $w_i \ge 0$,
and the set of centers must have \emph{total weight} at most $k$,
rather than \emph{size} at most $k$.
Does weighted $k$-Median have a polynomial-time
$O(\log (n/k))$-size approximation algorithm?
The result in~\cite{young_k_medians_2000} extends to this problem.

For the closely related Facilities Location problem,
the best polynomial-time approximation algorithm
returns a solution of cost at most
$c\cdot y^* + H_{\Delta} f\cdot x^*$,
where $(x^*, y^*)$ is an optimum solution to the standard LP relaxation for Facilities Location.
That is, it achieves ratio $H_\Delta$ with respect to the opening costs,
and ratio 1 with respect to the assignment costs.
This algorithm is relatively slow because it must solve the LP to obtain $(x^*, y^*)$.
Is there a faster greedy algorithm with the same performance guarantee?


\appendix

\section{Appendix}\label{sec: appendix}
\begin{proof}[Proof of \autoref{lemma: alpha}]
    \begin{align*}
        \alpha_{kn} = T/k+2
        & \le \ln(n^2/(2k(2k+1)) + 1/k + 2
        && (\text{by defn of } T)
        \\
        & = 2\ln(n/k)  + 2 - 2\ln 2 + 1/k + \ln(1 - 1/(2k+1))
        \\
        & \le 2\ln(n/k) + 2 - 2\ln 2 + 1/k - 1/(2k+1)
        && (\ln(1+z) \le z)
        \\
        & = 2\ln(n/k) + 2- 2\ln 2 + 1/(2k) + 1/(2k(2k+1)) && \qedhere
    \end{align*}
\end{proof}

\medskip

\begin{proof}[Proof of \autoref{lemma: submodular}]
    Consider the following random experiment, from~\cite{young_k_medians_2000}.
    Let $(x^*, y^*)$ be an optimal solution (of cost $\lambda^*$)
    to the $k$-Median LP relaxation (\autoref{fig: k median LP}).
    Choose a center $i'\in U$ randomly from the distribution $x^*/k$,
    then, for each customer $j\in W$ independently,
    reassign $j$ to $i'$ (in place of whatever current assignment it has)
    with probability $y^*_{i'j}/x_i$.

    As observed in~\cite{young_k_medians_2000},
    for any customer $j \in W $,
    the probability that $j$ is reassigned is $1/k$,
    and the expected cost of $j$'s new assignment, given that it is reassigned, is $\sum_{i=1}^n y^*_{ij} c_{ij}$,
    so
    \begin{align*}
        \min_{i\in U} \capped(\lambda, C\cup\{i\})
        & \le \E[\capped(\lambda, C\cup\{i'\})]
        = \sum_{j\in W} \E[\capped_j(\lambda, C\cup\{i\})]
        \\
        & \le
        \sum_{j\in W} \textstyle
        (1-1/k)\capped_j(\lambda, C)
        + (1/k)\sum_{i} y^*_{ij} (1-2k/n) c_{ij}/\lambda
        \\
        & =
        (1-1/k)\capped(\lambda, C)
        + (1/k)(1-2k/n) \lambda^*/\lambda.
        \qedhere
    \end{align*}

\end{proof}

\bibliography{99_bib}
\bibliographystyle{siamplain}

\end{document}